\documentclass[12pt,letterpaper]{article}
\pdfoutput=1

\usepackage{amsmath,amssymb, amsthm}
\usepackage{bbm}
\usepackage{calc, comment}
\usepackage{dsfont}
\usepackage{enumerate, enumitem, epsfig}
\usepackage{float}
\usepackage{graphicx}
\usepackage[linktocpage=true]{hyperref}
\usepackage[utf8]{inputenc} 
\usepackage{mathrsfs, mathtools}
\usepackage[numbers,sort&compress]{natbib}
\usepackage{physics, psfrag}
\usepackage{soul}
\usepackage{tikz}
\usepackage{tikz-cd}
\usepackage{tensor}
\usepackage[dvipsnames]{xcolor}

\newtheorem{theorem}{Theorem}
\newtheorem{lemma}{Lemma}
\newtheorem{corollary}{Corollary}

\definecolor{azure}{rgb}{0.0, 0.5, 1.0}
\definecolor{darkblue}{rgb}{0.15,0.35,0.7}
\definecolor{reddish}{rgb}{0.65, 0.2, 0.2}
\definecolor{brandeisblue}{rgb}{0.0, 0.44, 1.0}
\definecolor{ceruleanblue}{rgb}{0.16, 0.32, 0.75}
\definecolor{indigo(dye)}{rgb}{0.0, 0.25, 0.42}

\hypersetup{
colorlinks=true,
citecolor=ceruleanblue,
linkcolor=ceruleanblue,
urlcolor=ceruleanblue,
pdfauthor={},
pdftitle={},
pdfsubject={}
}

\DeclareFontEncoding{LS1}{}{}
\DeclareFontSubstitution{LS1}{stix}{m}{n}
\DeclareSymbolFont{stixsymbols}{LS1}{stixscr}{m}{n}
\SetSymbolFont{stixsymbols}{bold}{LS1}{stixscr}{b}{n}
\DeclareMathSymbol{\kay}{\mathalpha}{stixsymbols}{"6B}
\DeclareMathSymbol{\hay}{\mathalpha}{stixsymbols}{"68}

\makeatletter
\renewcommand\section{\@startsection {section}{1}{\z@}%
                               {-3.5ex \@plus -1ex \@minus -.2ex}%nn
                               {2.3ex \@plus.2ex}%
                               {\normalfont\large\bfseries}}
\renewcommand\subsection{\@startsection{subsection}{2}{\z@}%
                                 {-3.25ex\@plus -1ex \@minus -.2ex}%
                                 {1.5ex \@plus .2ex}%
                                 {\normalfont\bfseries}}
\makeatother

\makeatletter
\newcommand*\bigcdot{\mathpalette\bigcdot@{.5}}
\newcommand*\bigcdot@[2]{\mathbin{\vcenter{\hbox{\scalebox{#2}{$\m@th#1\bullet$}}}}}
\makeatother

\newfont{\goth}{ygoth.tfm scaled 1200}                   % gothic font (usual)
\numberwithin{equation}{section}

\allowdisplaybreaks
\begin{document}
%%%%%%%%%%%%%%%%
%%%%%%%%%%%%%%%%
\begin{titlepage}
\begin{flushright}
\today
\end{flushright}
\vspace{5mm}

\begin{center}
{\Large \bf 
Neural Network Field Theory at Finite Width}
\end{center}

\begin{center}

{\bf
Christian Ferko${}^{a, b}$ and
Aaron Mutchler${}^{a}$
} \\
\vspace{5mm}

\footnotesize{
${}^{a}$
{\it 
Department of Physics, Northeastern University, Boston, MA 02115, USA
}
 \\~\\
${}^{b}$
{\it 
The NSF Institute for Artificial Intelligence
and Fundamental Interactions
}
}
\vspace{4mm}
~\\
\texttt{c.ferko@northeastern.edu,
mutchler.a@northeastern.edu
}\\
\vspace{2mm}

\end{center}

\begin{abstract}
\baselineskip=14pt

Under mild assumptions, any quantum mechanical (QM) model or quantum field theory (QFT) admits a representation in terms of an ensemble of neural networks with countably many random parameters.  We investigate the features of NN-QM and NN-FT models with finitely many parameters, such as a feedforward network of width $N < \infty$. We find that, generically, such models must violate one of the properties of conventional Euclidean QFTs, such as reflection positivity or cluster decomposition. We present several complementary ways of understanding which features can and cannot be preserved at finite $N$, both in QM and in QFT.

\end{abstract}
\vspace{5mm}

\vfill
\end{titlepage}

%\tableofcontents

%\newpage
\renewcommand{\thefootnote}{\arabic{footnote}}
\setcounter{footnote}{0}

\tableofcontents{}
% \vspace{1cm}
\bigskip
\hrule

% \newpage

\section{Introduction}\label{sec:intro}

In many cases, the Euclidean path integral can be rigorously defined and provides a measure over the space of field configurations.
In such theories, correlation functions can then be formally viewed as certain expectation values taken with respect to this measure:
\begin{align}\label{path_integral_correlator}
    \langle \phi ( x_1 ) \ldots \phi ( x_n ) \rangle = \int \mathcal{D} \phi \, \phi ( x_1 ) \ldots \phi ( x_n ) = \mathbb{E} \left[ \phi ( x_1 ) \ldots \phi ( x_n ) \right] \, .
\end{align}
The idea of neural network field theory (NN-FT) \cite{Halverson:2020trp,Halverson:2021aot} is to ask whether this measure over field configurations in a physically realistic QFT can be obtained from an ordinary probability distribution over real variables called \emph{parameters}, which we indicate with the variable $\theta$. Given a draw of these real-valued parameters, one then obtains a field configuration $\phi_\theta ( x )$ by inserting the values of $\theta$ into a prescribed function known as the \emph{architecture}. With a probability distribution over parameters $ P ( \theta ) $, known as the \emph{parameter density}, the combined data $ \left( \phi_{ \theta } , P ( \theta ) \right)$ induces a measure over field configurations by pushing forward the density $P ( \theta )$ using the architecture $\phi_\theta$. Here and in the remainder of this work, we use a more general notion of ``neural network'' which refers to any such parameterized family of random functions $\phi_\theta ( x )$; when restricting to familiar architectures that are typically employed in machine learning, we will often adopt more specific terminology and speak of ``conventional'' or ``realistic'' architectures.

If a given quantum field theory has correlation functions (\ref{path_integral_correlator}) which agree with those computed from the neural network expectation values, i.e.
\begin{align}
    \langle \phi ( x_1 ) \ldots \phi ( x_n ) \rangle = \int d \theta \, P ( \theta ) \phi_\theta ( x_1 ) \ldots \phi_\theta ( x_n )
\end{align}
for all $n$ and all $x_i$, then we say that the data $( \phi_\theta , P ( \theta ) )$ provides a neural network realization of the quantum field theory.

Broad existence results about such neural network realizations are known. Under minimal assumptions, every one-dimensional quantum field theory -- which is simply a theory of quantum mechanics, for instance describing trajectories $x(t)$ of a quantum particle -- admits a neural network realization with a countable infinity of parameters \cite{Ferko:2025ogz}. This result was generalized in \cite{Ferko:2026axm}, where it was shown that any $d$-dimensional quantum field theory associated with a non-negative measure over field configurations likewise admits a neural network realization with a countable infinity of parameters.

The existence results above place neural networks with a countable infinity of random parameters at the foundation of NN-QM and NN-FT. Although these results establish the existence of neural network realizations under mild assumptions, they leave open a natural question: what happens when the countably infinite construction is replaced by a network with finitely many parameters? In particular, which properties of the underlying quantum theory can be realized exactly, which survive approximately, and which must be relaxed? We take this finite-$N$ question as the organizing perspective of the present work.

There is a subtlety, due to the Borel isomorphism theorem, in formulating this question, since the number of random parameters alone does not distinguish a genuinely finite network from an infinite construction. For the standard Borel spaces of paths or field configurations relevant here, any probability law can formally be obtained as the measurable pushforward of a single parameter $\theta \sim U ( [ 0 , 1 ] )$ through some architecture \cite{Ferko:2026axm}. Thus this one-parameter construction is compatible with the definition of an (exact) NN-QM or NN-FT realization given above. This observation was part of our motivation for the present investigation: it forces one to ask what additional structure is intended when such a realization is called finite. For a generic continuum theory, the single-parameter realization above encodes an unbounded amount of information in the infinitely many digits of $\theta$, and thus requires infinitely many operations to evaluate. A finite parameter count therefore need not describe a finite network in any computational or practical sense.

It is therefore important to distinguish architecture-agnostic statements from conclusions about particular classes of architectures. Some of the arguments in this article, as well as the related work of \cite{thomas}, are agnostic as to the specific functional form of the architecture, so long as it satisfies certain stated hypotheses. Much of our analysis, however, concerns natural finite-width architectures which cannot hide an infinite computation in a finite collection of parameters. Our representative example is the single-layer architecture
\begin{align}\label{single_layer}
    \phi_\theta ( x ) = \sum_{i=1}^{N} w_i \sigma_i \left( a_i \cdot x + b_i \right) \, ,
\end{align}
with finitely many parameters $\theta \in \{ w_i, a_i, b_i \}$. We restrict to a single layer for simplicity, although analogous questions may be posed for deep networks, and we allow the neurons $\sigma_i$ to vary for generality. Networks of the form (\ref{single_layer}) belong to the familiar universal-approximation class of Cybenko's theorem as the width is allowed to grow \cite{cybenko1989approximation}. However, we stress that this approximation for \emph{fixed functions} does not imply that such an architecture with random parameters can correctly reproduce all correlation functions of a quantum theory at any fixed finite $N$.

Our aim is therefore to characterize which features of conventional quantum theories can be reproduced by finite-$N$ architectures, and how the answer depends on the architecture and physical assumptions. One might immediately suspect that finite-$N$ architectures must sacrifice some of the typical properties of quantum systems, owing to the classes of objects on which typical path integral measures are supported. Here it is helpful to separate the case of $1d$ quantum field theories (quantum mechanics) from higher dimensional field theories.

\begin{enumerate}[label = (\roman*)]
    \item\label{qm_intuitive} \textbf{Quantum mechanics}. For $d = 1$, the standard path integral measure is supported on paths which are continuous everywhere but differentiable nowhere \cite{FeynmanHibbs1965,Simon2005Functional}. For instance, the path integral for the free quantum particle is supported entirely on Brownian paths. A finite sum of the form (\ref{single_layer}), where each of the activation functions $\sigma_i$ is piecewise differentiable with finitely many points of non-differentiability (e.g. the ReLU function) can never generate a path which is nowhere-differentiable.

    \item\label{qft_intuitive} \textbf{Quantum field theory}. In dimensions $d \geq 2$, the path integral measure is supported on proper Schwartz distributions, i.e. Schwartz distributions which are not themselves functions.\footnote{This was proven, for instance, for the massive free field \cite{ColellaLanford1973}, for a general class of $P(\phi)_2$ theories \cite{FrohlichSimon1977}, and for the sine--Gordon model on $\mathbb{T}^2$ with $0<\beta^2<4\pi$ \cite{OhRobertSosoeWang2021}.} Said differently, a typical field configuration $\phi ( x )$ in quantum field theory is not a function that can be evaluated at a point, but it can be integrated against appropriate test functions, much like the Dirac delta distribution. However, a finite sum (\ref{single_layer}) where each of the $\sigma_i$ is an ordinary function must also be a function and not a proper Schwartz distribution. 
\end{enumerate}

It turns out that the intuition of \ref{qm_intuitive} and \ref{qft_intuitive} is correct. However, we find it instructive to formalize these results for several reasons. First, it is useful to see which features are exhibited by finite-width approximations to a quantum mechanical model or quantum field theory. We will see that generically such models do not obey all of the Osterwalder-Schrader (OS) axioms \cite{Osterwalder:1973dx,Osterwalder:1974tc} -- in particular, reflection positivity (RP) and cluster decomposition play key roles -- but it is worthwhile to relax some implications of the OS axioms to uncover what is possible at finite $N$. Second, it is interesting to consider generalizations of architectures like (\ref{single_layer}) to see whether they can realize other physical properties. For instance, one might generalize (\ref{single_layer}) by allowing the activation functions $\sigma_i$ to be nowhere-differentiable but square-integrable (in quantum mechanics) or by allowing the $\sigma_i$ to be proper Schwartz distributions (in QFT), but still only retaining finitely many terms in the sum. It turns out that neither of these generalizations is sufficient: a finite-width network with nowhere-differentiable or distributional neurons still cannot reproduce all correlation functions of a fully-OS quantum mechanical theory or a quantum field theory, respectively.

This is not to say that finite-width neural network field theories are uninteresting. If one knows the correct infinite-width architecture and parameter density that realize a target QFT, then in principle, truncating this architecture to finite $N$ gives a way to approximate observables using Monte Carlo sampling. In general, one has no guarantees about the statistical properties of this approximation scheme, and we do not consider any such error analysis in this work. Nonetheless, such truncated approximations have been used to perform numerical NN-FT studies of Liouville theory \cite{Ferko:2026axm}, field theories with topological features \cite{Ferko:2026ken} or with Virasoro symmetry \cite{Robinson:2025ybg}, and Maxwell theory \cite{toappear}.

Besides offering a potential method for approximating observables, finite-width NN-FTs exhibit interesting features in their own right; for instance, truncating neural network realizations of free theories to finite $N$ often introduces non-local interactions that scale as $\frac{1}{N}$ \cite{Demirtas:2023fir,Robinson:2025ybg}. This non-locality offers an additional reason to expect that a finite neural network can never reproduce the correlation functions of a conventional (local) quantum field theory. However, although the Edgeworth expansion can be used to study these non-local interactions for finite-$N$ neural network representations of free field theories, to the best of our knowledge there is no general proof that finite-$N$ NN-FTs must be non-local in general. This again motivates the study of properties of finite-$N$ NN-QM and NN-FT models, including the results presented in the current work.

The remainder of this article is organized as follows. In Section \ref{sec:qm}, we recall some generalities about neural network quantum mechanics \cite{Ferko:2025ogz} (see also \cite{Hashimoto:2024aga} for related work) and then provide two complementary proofs that fixed-feature finite-width representations fail to give conventional QM: one which relies on the infinite-rank nature of the covariance kernel, and one which applies the optimality of the Kosambi-Karhunen-Lo\`eve \cite{kosambi1943statistics,karhunen1947,Loeve1948} expansion for stochastic processes. In Section \ref{sec:qft}, we turn to quantum field theory. We first review that the OS axioms of temperedness, Euclidean invariance, and reflection positivity generically imply that two-point functions diverge in coincident-point limits, and we give conditions on a class of finite-$N$ networks which cannot reproduce such divergences. We then discuss some examples of finite-width networks which \emph{can} produce these coincident $2$-point divergences (and RP two-point functions). In Section \ref{sec:momentum_space}, we further assume the cluster property, and show that a broad class of finite-$N$ networks necessarily violate at least one of the OS axioms. Section \ref{subsec:distributional_neurons} asks what is possible for finite-$N$ networks whose neurons are general Schwartz distributions, and the answer is that they still cannot reproduce the K\"all\'en-Lehmann decomposition implied by a subset of the OS axioms. In Section \ref{sec:conclusion}, we summarize our results and present directions for future investigation.

While this paper was being completed, we became aware of the interesting related manuscript \cite{thomas}, which contains complementary results about finite-width NN-FTs, and in particular the limitations of their usage for numerical simulations on a computer.

\section{Neural Network Quantum Mechanics at Finite \texorpdfstring{$N$}{TEXT}}\label{sec:qm}

We first introduce some generalities about neural network quantum mechanics models \cite{Ferko:2025ogz} which will be useful for the later discussion. As we mentioned, a QM model can be viewed as a one-dimensional QFT whose basic degree of freedom is a real-valued trajectory $x ( t )$, such as the position of a quantum particle, and the observables of interest are the Euclidean correlation functions
\begin{align}
    G^{(n)} ( t_1, \ldots, t_n ) = \left\langle x ( t_1 ) \cdots x ( t_n ) \right\rangle \, .
\end{align}
We trust that the use of the symbol $x$, which is here analogous to the field $\phi$ in QFT, will not be confused with its usage as a spacetime coordinate when writing $\phi ( x )$. Here the only spacetime coordinate is the Euclidean time $t$.

One can also view $x ( t )$ as a \emph{stochastic process}, and we sometimes use the notation $x_t$ for $x ( t )$ as is common in that literature. For every $t$ in an index set $T$, the quantity $x ( t )$ is a random variable, and the joint distributions of any finite collection of these random variables determine the process. We will work on a compact interval $T = [ a, b ]$ with $a<b$.\footnote{Any realization on the full Euclidean time axis $t \in \mathbb{R}$ can also be restricted to such an interval, so any obstruction on $T$ is already an obstruction to a global realization.} As in \cite{Ferko:2025ogz}, no path-integral action is assumed; the correlation functions themselves are the defining data.

For our purposes, a neural network is a parameterized family of functions $\phi_\theta : T \to \mathbb{R}$ together with a probability distribution (or law) $P ( \theta )$ for the parameters. A draw of $\theta$ produces a draw of the function $\phi_\theta$, and hence a stochastic process whose correlation functions are
\begin{align}
    G^{(n)} ( t_1, \ldots, t_n )
    = \int P ( d \theta ) \, \phi_\theta ( t_1 ) \cdots \phi_\theta ( t_n ) \, .
\end{align}
If these quantities agree with the correlation functions of a Euclidean quantum system for all $n$ and all choices of the $t_i$, then $( \phi_\theta , P )$ is a neural network realization of that quantum system, or an NN-QM. Said differently, an NN-QM pushes forward the parameter law to an ensemble of trajectories with the desired Euclidean correlation functions.

Not every stochastic process defines a conventional quantum theory. We therefore restrict attention to processes obeying the Osterwalder-Schrader axioms \cite{Osterwalder:1973dx,Osterwalder:1974tc}, together with the mild requirement that the observable $x$ have finite variance.\footnote{Although the finite-variance assumption is natural in $1d$, as we will review, for quantum field theories in $d \geq 2$ the corresponding variance is typically infinite.} It is useful to separate the deterministic one-point function from the fluctuations by defining
\begin{align}
    \mu ( t ) = \langle x ( t ) \rangle \, , \qquad
    \hat{x}_t = x_t - \mu ( t ) \, , \qquad
    C ( t, s ) = \left\langle \hat{x}_t \hat{x}_s \right\rangle \, .
\end{align}
Here $C ( t, s )$ is the connected two-point function, or covariance kernel. Euclidean time-translation invariance makes $\mu ( t )$ constant and $C ( t, s )$ a function of $| t-s |$. Reflection positivity then enters the Osterwalder-Schrader reconstruction theorem, which produces a quantum-mechanical Hilbert space with a non-negative Hamiltonian, after shifting the ground-state energy to zero. Applying the spectral theorem to this Hamiltonian gives the one-dimensional K\"all\'en-Lehmann representation \cite{Kallen:1952zz,Lehmann1954berEV},
\begin{align}\label{kl_representation}
    C ( t, s ) = \int_0^\infty \rho ( d m ) \, e^{-m | t-s |} \, ,
\end{align}
where $\rho$ is a finite, non-negative measure.\footnote{Cluster decomposition forces $\rho ( \{ 0 \} ) = 0$ for the connected covariance. For simplicity, we will thus be cavalier about possible point contributions to the KL representation at $m = 0$, both here and in Section \ref{sec:qft}, although such contributions do not affect our results.} The representation (\ref{kl_representation}) and its higher-dimensional analogue will play a critical role throughout this work. We will call the positive-mass part of this representation non-trivial when
\begin{align}
    \rho \big( ( 0, \infty ) \big) > 0 \, .
\end{align}
This condition excludes a trivial stochastic process whose (centered) fluctuations consist only of a time-independent random mode, which corresponds to a measure supported entirely at $m=0$. A delta-function measure supported at a positive value of $m$, such as the one associated with the quantum harmonic oscillator, is non-trivial in this sense. Both here and in the QFT discussion of Section \ref{sec:qft}, we will ignore possible OS-null contact terms.

The representation (\ref{kl_representation}) implies that the stochastic process $x_t$ obeys the assumptions of the Kosambi-Karhunen-Lo\`{e}ve (KKL) theorem \cite{kosambi1943statistics,karhunen1947,Loeve1948}. The covariance kernel defines an operator $\mathcal{O}$ on $L^2 ( T )$ by
\begin{align}\label{O_operator_defn}
    \left( \mathcal{O} f \right) ( t ) = \int_a^b d s \, C ( t, s ) f ( s ) \, .
\end{align}
The eigenfunctions and eigenvalues of this linear operator furnish us with a neural network representation of the stochastic process $x_t$ with a countable infinity of parameters. This is the content of the following Lemma, which is a variant of the argument presented in \cite{Ferko:2025ogz}.

\begin{lemma}\label{kl_implies_kkl}
    Let $x_t$ be a stochastic process on $T = [ a, b ]$ whose covariance has the form (\ref{kl_representation}) for a finite, non-negative measure $\rho$. Then $x_t$ admits a decomposition
    \begin{align}\label{kkl_decomposition}
        x_t = \mu ( t ) + \sum_{k=1}^{\infty} \theta_k e_k ( t ) \, ,
    \end{align}
    where the $e_k$ are continuous orthonormal eigenfunctions of $\mathcal{O}$ associated with its non-zero eigenvalues $\mathcal{O} e_k = \lambda_k e_k$. We list these eigenvalues in non-increasing order, with each $\lambda_k > 0$, and the $\theta_k$ are pairwise uncorrelated random variables satisfying
    \begin{align}\label{thetas_to_lambdas}
        \langle \theta_k \rangle = 0 \, , \qquad
        \langle \theta_j \theta_k \rangle = \lambda_k \delta_{j k} \, .
    \end{align}
\end{lemma}
The content of Lemma \ref{kl_implies_kkl} is that every quantum model which admits a K\"all\'en-Lehmann representation also admits a single-layer neural network representation, where we view the $e_k ( t )$ as fixed deterministic neurons or features and $\theta_k$ as the random parameters. When we say that the sum (\ref{kkl_decomposition}) represents the QM model or stochastic process, as above, we mean that it reproduces its correlation functions. Mathematically, one also has that the KKL sum converges in mean square for every $t \in T$, i.e.
\begin{align}
    \lim_{M \to \infty} \left\langle \left| \hat{x}_t - \sum_{k=1}^{M} \theta_k e_k ( t ) \right|^2 \right\rangle = 0 \, .
\end{align}

\begin{proof}
    We first verify the regularity needed for the KKL expansion. Since
    $0 \leq e^{-m | t-s |} \leq 1$ and $\rho$ is finite, the dominated convergence theorem
    shows that $C ( t, s )$ is continuous on $T \times T$. Furthermore,
    \begin{align}
        C ( t, t )
        &= \rho \big( [ 0, \infty ) \big) < \infty \, , \nonumber \\
    \end{align}
    and
    \begin{align}
        \left\langle \left| \hat{x}_t - \hat{x}_s \right|^2 \right\rangle
        &= 2 \int_0^\infty \rho ( d m )
        \left( 1 - e^{-m | t-s |} \right)
        \longrightarrow 0
    \end{align}
    as $t \to s$. Thus the centered process $\hat{x}_t$ is zero-mean,
    square-integrable, and continuous in mean square.

    Because $C$ is a continuous covariance kernel on a compact interval,
    the associated operator $\mathcal{O}$ of (\ref{O_operator_defn}) is a compact, self-adjoint, positive operator on $L^2 ( T )$.
    The spectral theorem therefore gives an at most countable collection of
    non-zero eigenvalues $\lambda_k > 0$, which we list in non-increasing order,
    together with corresponding orthonormal eigenfunctions $e_k$. These
    eigenfunctions have continuous representatives, since the eigenvalue
    equation gives
    \begin{align}
        e_k ( t )
        = \frac{1}{\lambda_k}
        \int_a^b d s \, C ( t, s ) e_k ( s ) \, .
    \end{align}
    Mercer's theorem\footnote{Recall that Mercer's theorem upgrades the spectral decomposition of a continuous, symmetric, positive-semidefinite kernel on a compact domain to a pointwise eigenfunction expansion that converges absolutely and uniformly.} then yields the covariance expansion
    \begin{align}
        C ( t, s )
        = \sum_{k=1}^{\infty}
        \lambda_k e_k ( t ) e_k ( s ) \, .
    \end{align}
    Here and below, the sums run over the non-zero modes and are understood to
    terminate if $\mathcal{O}$ has finite rank.

    We now define the random coefficients by projecting the centered process
    onto these eigenfunctions,
    \begin{align}
        \theta_k
        = \int_a^b d s \, e_k ( s ) \hat{x}_s \, ,
    \end{align}
    where the integral is understood in the mean-square sense.
    
    Orthonormality together with $\mathcal{O} e_k = \lambda_k e_k$ gives
    \begin{align}
        \left\langle \theta_k \right\rangle &= 0 \, ,
    \end{align}
    along with
    \begin{align}
        \left\langle \theta_j \theta_k \right\rangle
        &= \int_a^b d t \int_a^b d s \,
        e_j ( t ) C ( t, s ) e_k ( s ) \nonumber \\
        &= \int_a^b d t \,
        e_j ( t ) \left( \mathcal{O} e_k \right) ( t ) \nonumber \\
        &= \lambda_k \delta_{j k} \, .
    \end{align}
    This proves (\ref{thetas_to_lambdas}). In the same way,
    \begin{align}
        \left\langle \hat{x}_t \theta_k \right\rangle
        = \int_a^b d s \, C ( t, s ) e_k ( s )
        = \lambda_k e_k ( t ) \, .
    \end{align}
    It follows that the mean-square error of the first $M$ modes is
    \begin{align}
        \left\langle
        \left|
        \hat{x}_t - \sum_{k=1}^{M} \theta_k e_k ( t )
        \right|^2
        \right\rangle
        = C ( t, t )
        - \sum_{k=1}^{M} \lambda_k e_k ( t )^2
        \longrightarrow 0 \, ,
    \end{align}
    where the final step follows from the diagonal part of the Mercer
    expansion. Thus the sum converges to $\hat{x}_t$ in mean square for every
    $t \in T$. Adding back the deterministic mean $\mu ( t )$ proves
    (\ref{kkl_decomposition}).
\end{proof}

We now ask whether the infinite neural network representation (\ref{kkl_decomposition}) can be replaced exactly by a finite-$N$ single layer network. By this we mean an architecture
\begin{align}\label{finite_kkl_like}
    x_t = \phi_\vartheta ( t ) = x_0 ( t ) + \sum_{k=1}^{N} \vartheta_k f_k ( t ) \, , \qquad N < \infty \, ,
\end{align}
where $x_0, f_k \in L^2 ( T )$ are fixed deterministic functions and the $\vartheta_k$ are random variables with an arbitrary joint distribution. The functions $f_k$ are completely general square-integrable features, and we impose no differentiability assumption upon them, so the argument below does not rely on using smooth or piecewise differentiable activations. We note that this is a less general architecture than the $1d$ version of (\ref{single_layer}), which also permits random input weights, whereas here we neglect this possibility and restrict to fixed features for simplicity. We will return to the more general single-layer architecture in Section \ref{sec:qft}, and also briefly comment on a quantum mechanical example in Section \ref{sec:counterexample}.

\begin{theorem}[Finite-width NN-QM obstruction]\label{QMdimensionalThm}
    Let $x_t$ be a real stochastic process on an interval $T = [ a, b ]$ whose connected two-point function admits the K\"all\'en-Lehmann representation (\ref{kl_representation}) with non-trivial measure $\rho$. Then no finite-width architecture of the form (\ref{finite_kkl_like}) can reproduce the connected two-point function of $x_t$. In particular, such an architecture cannot provide an exact neural network realization of $x_t$.
\end{theorem}

Because the Osterwalder-Schrader assumptions furnish (\ref{kl_representation}), the theorem applies in particular to any OS quantum system. The conclusion is already a statement about the two-point function and therefore requires neither Gaussianity nor any assumptions about higher moments. We will prove it in two complementary ways. In Subsection \ref{nnqmDimCount}, we show that the K\"all\'en-Lehmann covariance has strictly positive variance in every non-zero direction of $L^2 ( T )$, whereas a finite feature space necessarily has an orthogonal direction with zero variance. We then give a second proof using the optimality of the KKL basis: its covariance operator has infinitely many strictly positive eigenvalues, so every finite KKL truncation has non-zero mean-square error, while an exact finite representation would have zero error. Let us begin with the dimension-counting argument.

\subsection{Dimension-Counting}\label{nnqmDimCount}
At finite $N$, we find the NN representation of QM breaks down. Here we show the breakdown as caused by a mismatch in dimensions. Take a linear finite NN,

\begin{equation}\label{finiteNNQM}
    x_{ t } = x_{ 0 } ( t )\, +\, \sum_{ k = 1 }^{ N } \vartheta_{ k } f_{ k } ( t ),\quad N < \infty,
\end{equation}

\noindent where $x_0, f_k : T \to \mathbb{R}$ are deterministic measurable functions with $x_0, f_k \in L^2(T)$, and $\vartheta_k$ are random. The functions $f_k$ are also known as \emph{features} of the network. Every possible path, $x_t$, lives in a finite dimensional vector space, $x_t \in V = x_0 + \text{span}\left \{f_1, \ldots, f_N \right \}$. Formally, it is a random variable on $L^2(T)$ and for every parameter draw there is a natural pairing,

\begin{equation}\label{qmSMEAR}
    \langle x_t, h \rangle_{ L^2 ( T ) } = \int_T x_t h( t )\, dt\,,\ \text{ for all } h \in L^{2} ( T )\, .    
\end{equation}

\noindent Equation (\ref{qmSMEAR}) describes how much the process $x_t$ has moved along $h$ and is itself a random variable. The following lemma describes the variance of how well the process aligns with any $h$.

\begin{lemma}\label{positiveQMCovariance}
    Let $x_t$ be a stochastic process on $T=[a,b]$ whose connected covariance has the K\"all\'en-Lehmann representation (\ref{kl_representation}), where $\rho$ is a non-trivial measure. Then for every non-zero $ h \in L^2 ( T, \mathbb{R} ) $, the random variable
    \begin{equation}
        \hat{ x } ( h ) = \int_{T} \hat{ x }_t \, h ( t ) \, dt\, ,
    \end{equation}
    \noindent has strictly positive variance, $ \langle \hat{x} ( h )^2 \rangle  > 0 $.
\end{lemma}

\begin{proof}
    We show this by directly looking at the variance of $\hat { x } ( h ) $,
    \begin{equation*}
        \langle \hat{ x } ( h )^2 \rangle = \left \langle \int \hat{x}_{t} \hat{x}_s\, h(t) h(s)\, dt\, ds \right \rangle = \int \left \langle \hat{x}_t \hat{x}_s\right \rangle\,  h(t) h(s)\, dt\, ds \, .
    \end{equation*}

    \noindent The function $h$ is deterministic, while the path $\hat{x}_t$ is random. Using a known Fourier identity,
    \begin{equation*}
        e^{ -m \vert t - s \vert } = \frac{1}{ 2 \pi } \int \frac{ 2 m }{ \omega^2 + m^2 }\, e^{ i \omega ( t - s )} \, d\omega\, ,
    \end{equation*}
    we can expand the KL spectral representation as
    \begin{align*}
        \int \left \langle \hat{x}_t \hat{x}_s\right \rangle\,  h(t) h(s)\, dt\, ds &= \int e^{-m \vert t - s \vert }  h ( t ) h ( s )\, \rho(\, dm\, )\, dt \, ds\\
        &= \frac{ 1 }{ 2 \pi } \int \frac{ 2 m }{ \omega^2 + m^2 }\, e^{ i \omega ( t - s )}\, h ( t ) h ( s )\, \rho(\, dm\, )\,  d\omega\, dt\, ds\\
        &= \frac{ 1 }{  2 \pi  } \int \frac{ 2 m }{ \omega^2 + m^2 }\, \tilde{ h } ( \omega )  \tilde{ h } ( - \omega ) \, \rho(\, dm\, )\,  d\omega\\
        &= \frac{ 1 }{ 2 \pi } \int \frac{ 2 m }{ \omega^2 + m^2 }\, \vert \tilde{ h } ( \omega ) \vert^2 \rho(\, dm\, )\,  d\omega\, > 0\,.
    \end{align*}
    In going to the third equality we recognize the Fourier transform of $h$. Since $ h \neq 0 $, Plancherel's theorem tells us $ \tilde{ h } \neq 0 $. One need not worry about $h$ only defined on $T$, as we can extend it by 0 on $ T^{ c } $. Now, since everything in the integrand is positive, and by assumption the measure is positive, we have the desired claim, $ \langle \hat{ x } ( h )^2 \rangle > 0 $.
\end{proof}

A true NN representation of QM is supported on the infinite dimensional function space of paths that are continuous everywhere, differentiable nowhere. The paths in $V$ of a finite NN will never be able to fluctuate throughout the whole of $L^2(T)$, which is needed in order to capture the quantum properties of the representation. Herein lies the dimensional mismatch at finite $N$.

\begin{proof}[Proof of Theorem \ref{QMdimensionalThm}]
    Assume $x_t \in L^2(T)$ and is of the form equation (\ref{finiteNNQM}). Since it satisfies the OS axioms it must have a KL spectral presentation and therefore Lemma \ref{positiveQMCovariance} holds.\\
    \indent However, since $x_t \in V$ and $V$ is finite dimensional, there is a non-trivial $h \in L^2(T)$ that is orthogonal to all $f_k$. The functional $x(h)$ must have zero variance since
    \begin{equation*}
        x(h) = \int_T x_0(t)h(t)\, dt\,, 
    \end{equation*}
    is deterministic for \emph{every} parameter draw. Therefore $\operatorname{Var}\!\left(\phi_{\vartheta}(h)\right)=0$, and we have contradicted Lemma \ref{positiveQMCovariance}. Hence, no linear finite network can realize QM.
\end{proof}

We have thus shown that a single layer, finite-width network (\ref{finite_kkl_like}) with fixed features cannot reproduce a quantum mechanical stochastic process. The OS axioms require a K\"all\'en-Lehmann representation for the two-point function, which can only be achieved by $f_k$ that span the whole of $L^2(T)$. Physically, one might interpret this as the statement that the path integral instructs us to sum over \emph{all} possible paths, which cannot be achieved by retaining only a finite set of paths $f_k ( t )$.

This argument relied on the ability to prescribe a KL spectral representation to the two point function. By Lemma \ref{kl_implies_kkl}, the process also has a KKL decomposition. However, the KKL decomposition is a sum over infinitely many $e_k$. In the next section we explore the consequences of a finite truncation of the KKL expansion.

\subsection{KKL Optimality}

We will now present a different argument for Theorem \ref{QMdimensionalThm} which relies upon the fact that, among all orthogonal expansions which represent a stochastic process, the KKL expansion is the optimal one in the sense of minimizing truncation error. To review this optimality statement, let $u_1,\ldots,u_N$ be any orthonormal family in $L^2 ( T,\mathbb{R} )$, and let
\begin{align}
    V_N = \operatorname{span}
    \left\{ u_1,\ldots,u_N \right\} \, .
\end{align}
Among all approximations taking values in $V_N$, the one with the smallest integrated mean-square error is the orthogonal projection of $\hat{x}$ onto $V_N$. Its coefficients are
\begin{align}
    \eta_k = \hat{x} ( u_k ) = \int_a^b d t \, \hat{x}_t u_k ( t ) \, .
\end{align}
It is convenient to introduce the symbol $\mathcal{E}_N ( V_N )$ for this minimized mean-square truncation error associated with approximating a centered process $\hat{x}_t$ using the vector space $V_N$, i.e.
\begin{align}\label{error_defn}
    \mathcal{E}_N ( V_N ) =  \inf_{\eta_1,\ldots,\eta_N} \left\langle \left\| \hat{x} - \sum_{k=1}^{N} \eta_k u_k \right\|_{L^2 ( T )}^2 \right\rangle \, .
\end{align}
The result of KKL optimality is that this truncation error is bounded by a sum of eigenvalues associated with the operator $\mathcal{O}$ of (\ref{O_operator_defn}) used in constructing the KKL expansion:
\begin{align}\label{kkl_optimality}
    \mathcal{E}_N ( V_N ) \geq \sum_{k>N} \lambda_k \, .
\end{align}
The inequality (\ref{kkl_optimality}) is saturated by choosing
\begin{align}
    u_k=e_k,\qquad \eta_k=\theta_k
\end{align}
for which
\begin{align}
    V_N=\operatorname{span}\left\{e_1,\ldots,e_N\right\} \, ,
\end{align}
which means that one has used the truncated KKL expansion itself.\footnote{We neglect possible degeneracies that can occur when multiple KKL eigenvalues coincide.} Thus the first $N$ KKL modes capture the greatest possible variance among all $N$-dimensional deterministic feature spaces and minimize the integrated mean-square truncation error.

The positivity result of Lemma \ref{positiveQMCovariance} now has the following immediate consequence.

\begin{corollary}\label{no_truncation_lemma}
    Let $x_t$ obey the assumptions of Lemma \ref{kl_implies_kkl} with a non-trivial K\"all\'en-Lehmann measure $\rho$. Then $\mathcal{O}$ has infinite rank, and the KKL decomposition (\ref{kkl_decomposition}) contains infinitely many random coefficients, each with strictly positive variance:
    \begin{align}\label{kkl_coefficients_positive}
        \left\langle \theta_k^2 \right\rangle
        = \lambda_k
        > 0 \, .
    \end{align}
\end{corollary}

\begin{proof}
    We begin by applying Lemma \ref{positiveQMCovariance} where we choose the smearing function $h$ to be one of the KKL eigenfunctions $e_k$. Using orthonormality, the resulting smeared random variable is simply
    \begin{align}
        \theta_k = \hat{x} ( e_k ) = \int_a^b d t \, \hat{x}_t e_k ( t ) \, ,
    \end{align}
    and thus (\ref{thetas_to_lambdas}) and Lemma \ref{positiveQMCovariance} give
    \begin{align}
        \left\langle \theta_k^2 \right\rangle = \lambda_k > 0 \, .
    \end{align}
    As usual, we neglect a possible contribution to the measure $\rho$ at $m = 0$, which does not affect this conclusion. More generally, the same lemma implies
    \begin{align}
        \left\langle
            h,\mathcal{O}h
        \right\rangle_{L^2 ( T )}
        =
        \left\langle \hat{x} ( h )^2 \right\rangle
        > 0
    \end{align}
    for every non-zero $h \in L^2 ( T,\mathbb{R} )$. Hence $\ker\mathcal{O}=\{0\}$. Since $L^2(T)$ is infinite-dimensional, $\mathcal{O}$ cannot have finite rank. Its compactness, self-adjointness, and positivity therefore imply that it has infinitely many non-zero eigenvalues $\lambda_k>0$; the corresponding KKL coefficients satisfy $\langle\theta_k^2\rangle=\lambda_k>0$.
\end{proof}

An exact rank-$N$ representation would have zero integrated mean-square error after projection onto its feature space. Corollary \ref{no_truncation_lemma} and the optimality relation (\ref{kkl_optimality}) show that this is impossible. The following proof makes the contradiction explicit.

\begin{proof}[Proof of Theorem \ref{QMdimensionalThm}]
    Suppose, by way of contradiction, that an architecture of the form (\ref{finite_kkl_like}) reproduces the connected two-point function of $x_t$. For simplicity, assume that we have subtracted off the mean value and work instead with the centered version of the process. Without loss of generality, we may take the functions $f_k$ to be linearly independent; otherwise, replace the feature list with a basis for its span. We may then orthonormalize them using the Gram-Schmidt procedure. Explicitly, set
    \begin{align}
        g_1 = f_1 \, , \qquad \hat{f}_1 = \frac{g_1}{ \| g_1 \| } \, ,
    \end{align}
    where $\| f \|^2 = \langle f, f \rangle = \int_a^b dt \, f^2$ for any function $f: [ a, b ] \to \mathbb{R}$, and for $2 \leq k \leq N$ let
    \begin{align}
        g_k = f_k - \sum_{j=1}^{k-1} \langle f_k , \hat{f}_j \rangle \hat{f}_j \, , \qquad \hat{f}_k = \frac{g_k}{\| g_k \|} \, .
    \end{align}
    The stochastic process $x_t$ may therefore be expressed as
    \begin{align}
        x_t = \phi_\vartheta ( t ) = \sum_{k=1}^{N} \hat{\vartheta}_k \hat{f}_k ( t ) \, ,
    \end{align}
    where the $\hat{f}_k$ have been defined above, and the $\hat{\vartheta}_k$ are corresponding linear combinations of the original parameters $\vartheta_k$ which are defined by the condition that $\sum_{k=1}^{N} \vartheta_k f_k ( t ) = \sum_{k=1}^{N} \hat{\vartheta}_k \hat{f}_k ( t )$. The functions $\hat{f}_k$ are now orthonormal, but do not form a complete basis for the Hilbert space $H = L^2 ( [ a, b ] )$ because there are only finitely many of them. Extend this to a complete orthonormal basis as follows: let
    \begin{align}
        V = \mathrm{span} \left( \hat{f}_1 , \ldots, \hat{f}_N \right) \, , \qquad V^\perp = \{ h \in H \mid \langle h , v \rangle = 0 \text{ for all } v \in V \} \, ,
    \end{align}
    so that $H = V \oplus V^\perp$. Choose a basis $f_{N+1}, f_{N+2}, \ldots$ for $V^\perp$, and then again orthonormalize to find corresponding functions $\hat{f}_{N+1}$, $\hat{f}_{N+2}$, $\ldots$. We now have a complete orthonormal basis $\{ \hat{f}_k \mid k = 1, 2, \ldots \}$ for $H$, where the first $N$ of the $\hat{f}_k$ coincide with those defined above.

    By the assumption that $x_t$ admits a finite linear NN representation, we then have
    \begin{align}\label{truncated_orthonormal_contradiction}
        x_t = \phi_\vartheta ( t ) = \sum_{k=1}^{N} \hat{\vartheta}_k \hat{f}_k ( t ) = \sum_{k=1}^{\infty} \hat{\vartheta}_k \hat{f}_k ( t ) \, , 
    \end{align}
    where the random variables $\hat{\vartheta}_k$ vanish identically for $k > N$. In other words, using the orthonormal basis $\hat{f}_k$ and the random variables $\hat{\vartheta}_k$, there is no truncation error from terminating the infinite sum in the final expression of (\ref{truncated_orthonormal_contradiction}) at $k = N$. Compare this to the expansion of $x_t$ using the KKL decomposition, which is guaranteed to exist by Lemma \ref{kl_implies_kkl}:
    \begin{align}\label{xt_kkl_comparison}
        x_t = \sum_{k=1}^{\infty} \theta_k e_k ( t ) \, .
    \end{align}
    Truncating the sum (\ref{xt_kkl_comparison}) at the same upper value $N$ of the index $k$ incurs a non-zero truncation error, since by Corollary \ref{no_truncation_lemma}, all of the random variables $\theta_k$ have strictly positive variance. But this contradicts the fact that, among all orthonormal expansions of a stochastic process, the KKL decomposition minimizes the truncation error.
\end{proof}

The two proofs of Theorem \ref{QMdimensionalThm} reveal the same obstruction, but from complementary perspectives. The dimension-counting argument is qualitative: it identifies a non-zero direction orthogonal to the finite feature space in which the network has zero variance, contradicting the strict positivity of the K\"all\'en-Lehmann covariance. The KKL argument is instead variational and quantitative. For any rank-$N$ approximation within the finite-feature class considered here, its integrated mean-square error satisfies $\mathcal{E}_N \geq \sum_{k>N}\lambda_k \geq \lambda_{N+1}$. Thus the first omitted KKL eigenvalue supplies a basis-independent lower bound on the error of every finite-$N$ approximation, while the full collection of truncated eigenvalues gives the sharp optimal bound attained by the KKL truncation itself.

\section{Neural Network Quantum Field Theory at Finite \texorpdfstring{$N$}{N}}\label{sec:qft}

We now consider quantum field theories in $d \geq 2$ dimensions. As is standard in constructive QFT (see e.g. \cite{glimm2012quantum}), we phrase this discussion in terms of generalized functions or distributions (like the Dirac delta distribution), not to be confused with probability distributions.\footnote{See \cite{https://doi.org/10.1002/prop.2190440204} for a historical account of the discovery of the fact that quantum fields are distributional.} 

Let us make a brief mathematical remark in passing. In probabilistic language, the \emph{law} of a random field is the probability measure that it induces on the space of field configurations. Thus one could adopt a stronger notion of an NN-FT realization in which the joint laws of all finite collections of smeared field variables agree between the QFT and the neural network ensemble. Equality of these laws implies agreement of all smeared correlation functions, whereas the converse need not hold without additional moment-determinacy assumptions. For simplicity, however, we will generally avoid speaking about the law of random fields and rather speak about a quantum field theory as a probability distribution over Schwartz distributions. We trust that the difference between a probability distribution and a Schwartz distribution will be clear from context.

If $C^{\infty} ( \mathbb{R}^d )$ is the space of all infinitely differentiable functions on $\mathbb{R}^d$, let $\mathcal{S} ( \mathbb{R}^d ) \subset C^{\infty} ( \mathbb{R}^d )$ be the \emph{Schwartz space} of rapidly decaying test functions. More precisely, $\mathcal{S} ( \mathbb{R}^d )$ is the space of all smooth functions $f$ with the property that the product of any derivative of $f$, multiplied by any power of $|x|$, converges to $0$ as $x \to \infty$. By ``any derivative'' we include derivatives of arbitrarily high order and with respect to any of the $d$ arguments of $f$. Let $\mathcal{S}' ( \mathbb{R}^d )$ be the dual of the space of Schwartz functions, i.e. the space of continuous linear functionals on $\mathcal{S} ( \mathbb{R}^d )$. We refer to $\mathcal{S}' ( \mathbb{R}^d )$ as the space of (tempered) Schwartz distributions.\footnote{In some references, the term ``Schwartz distribution'' is reserved for linear functionals on the space of smooth functions with \emph{compact} support, whereas the phrase ``tempered distribution'' is used for the space which we call $\mathcal{S}' ( \mathbb{R}^d )$. For simplicity, we will neglect this distinction and simply use the single phrase ``Schwartz distribution'' for an element of $\mathcal{S}' ( \mathbb{R}^d )$ in what follows.} An element $\Phi \in \mathcal{S} ' ( \mathbb{R}^d )$ is a function which maps each $f \in \mathcal{S} ( \mathbb{R}^d )$ to a number $\Phi ( f )$.

Ordinary Schwartz functions may also be regarded as Schwartz distributions, where the linear pairing is provided by integration. For instance, if $\phi \in \mathcal{S} ( \mathbb{R}^d )$, one may apply the inclusion map $\iota : \mathcal{S} ( \mathbb{R}^d ) \to \mathcal{S}' ( \mathbb{R}^d )$ to obtain a distribution $\iota \phi = \Phi \in \mathcal{S} ' ( \mathbb{R}^d )$. The values of $\Phi$ on test functions are given by
\begin{align}\label{function_to_distribution}
    \Phi ( f ) = \int d^d x \, \phi ( x ) f ( x ) \, .
\end{align}
As we mentioned in the introduction, it will be convenient to speak of the space of proper Schwartz distributions, which are those elements $\Phi \in \mathcal{S} ' ( \mathbb{R}^d )$ such that there exists no function $h : \mathbb{R}^d \to \mathbb{R}$ with the property that $\Phi ( f ) = \int d^d x \, h ( x ) f ( x )$ for all $f \in \mathcal{S} ( \mathbb{R}^d )$. That is, proper Schwartz distributions are the distributions that \emph{cannot} be represented by ordinary functions. Note that the space of proper Schwartz distributions is not the same as $\mathcal{S}' ( \mathbb{R}^d ) \setminus \iota \left( \mathcal{S} ( \mathbb{R}^d ) \right)$ since there exist Schwartz distributions associated with functions that are not Schwartz functions. A simple example is any constant function on $\mathbb{R}^d$, which can be integrated against any test function but which does not decay at infinity.

Although Schwartz functions and their associated distributions are fundamentally different mathematical objects -- whereas $\phi : \mathbb{R}^d \to \mathbb{R}$ takes spacetime points as inputs, its associated distribution $\Phi : \mathcal{S} ( \mathbb{R}^d ) \to \mathbb{R}$ takes Schwartz functions as inputs -- we will be cavalier about the distinction between $\phi$ and $\Phi = \iota \phi$ in what follows, writing both as $\phi$.

Because typical field configurations in QFT take values in $\mathcal{S}' ( \mathbb{R}^d )$ and need not admit values at points, for clarity we will usually speak of correlation functions
\begin{align}
    G^{(n)} ( f_1, \ldots, f_n ) = \langle \phi ( f_1 ) \ldots \phi ( f_n ) \rangle
\end{align}
where $f_i \in \mathcal{S} ( \mathbb{R}^d )$, rather than referring to correlators such as $\langle \phi ( x_1 ) \phi ( x_2 ) \rangle$.

The relationship between these two perspectives can be established by considering so-called \emph{approximations to the identity} which converge to Dirac delta distributions in a suitable sense. Choose any non-negative smooth function $\eta ( x )$ with compact support and $\int d^d x \, \eta ( x ) = 1$. Then one can define the \emph{mollifier}
\begin{align}\label{mollifier_defn}
    \eta_{\epsilon, x} ( y ) = \epsilon^{-d} \eta \left( \frac{y-x}{\epsilon} \right) \, .
\end{align}
For instance, $\eta ( x )$ can be chosen to be a suitably normalized bump function. Then $\eta_{\epsilon, x} ( y ) \in \mathcal{S} ( \mathbb{R}^d )$ for any finite $\epsilon$, and the mollified field value $\phi ( \eta_{\epsilon, x} )$ takes a finite value for any $\phi \in \mathcal{S} ' ( \mathbb{R}^d )$. As $\epsilon \to 0$, the mollifier $\eta_{\epsilon, x} ( y )$ behaves as a delta function supported at the point $x = y$. Thus one can recover the familiar Schwinger functions $S^{(n)} ( x_1, \ldots, x_n )$, which take $n$ spacetime points as inputs, via the limiting procedure
\begin{align}
    S^{(n)} ( x_1 , \ldots, x_n ) &= \lim_{\epsilon \to 0} G^{(n)} \left( \eta_{\epsilon, x_1} , \ldots, \eta_{\epsilon, x_n} \right) = \lim_{\epsilon \to 0} \langle \phi ( \eta_{\epsilon, x_1} )  \ldots \phi ( \eta_{\epsilon, x_n} ) \rangle \, .
\end{align}
A quite general, non-perturbative fact about quantum field theories -- which requires only the OS axioms of temperedness, Euclidean invariance, and reflection positivity -- is the existence of a K\"all\'en-Lehmann representation for the two-point function. Although this is often written in the momentum-space form
\begin{align}\label{momentum_space_KL}
    G ( p ) = \int_0^\infty \frac{\rho ( d \mu^2 ) }{p^2 + \mu^2}
\end{align}
for a non-negative measure $\rho$, we will also work with the distributional form,
\begin{align}\label{main_kallen_lehmann}
    C ( f, g ) &= \mathbb{E} \left[ \overline{\phi ( f )} \phi ( g ) \right] \nonumber \\
    &= \int_0^\infty \rho ( d \mu^2 ) \int_{\mathbb{R}^d} \frac{d^d p}{ ( 2 \pi )^d } \frac{\overline{\widehat{f} ( p )} \widehat{g} ( p )}{p^2 + \mu^2} \, , \qquad \rho \geq 0 \, ,
\end{align}
where $\widehat{f}$ and $\widehat{g}$ denote the Fourier transforms of the test functions $f$ and $g$, respectively. If the measure $\rho$ vanishes identically, we refer to (\ref{main_kallen_lehmann}) as a \emph{trivial} K\"all\'en-Lehmann representation; the resulting theory likewise has a trivial two-point function which does not diverge at coincident points. In what follows, we will restrict to \emph{non-trivial} K\"all\'en-Lehmann representations, although we will sometimes omit this qualification for brevity. Likewise, when considering higher-point functions, we will also implicitly ignore the trivial quantum field theory whose $n$-point functions are all independent of position.

Here and below, the K\"all\'en-Lehmann representation is understood to apply to the connected two-point function, or equivalently to the centered field \(\phi-\langle\phi\rangle\). In the absence of cluster decomposition, the connected two-point function may also contain a position-independent zero-momentum contribution,
\begin{align}
    C_c(f,g)
    =
    c_0\,
    \overline{\widehat f(0)}\,\widehat g(0)
    +
    \int_0^\infty \rho(d\mu^2)
    \int_{\mathbb R^d}
    \frac{d^dp}{(2\pi)^d}
    \frac{\overline{\widehat f(p)}\,\widehat g(p)}
         {p^2+\mu^2},
    \qquad c_0\geq0 .
\end{align}
A random constant field \(\phi(x)=A\) has \(c_0=\operatorname{Var}(A)\) and \(\rho=0\), and is therefore excluded from this class. Regardless, since \(\widehat{\eta}_{\epsilon,x}(0)=1\), the term proportional to \(c_0\) contributes only a finite constant to the mollified coincident variance and does not affect the divergences considered below. Cluster decomposition would further require \(c_0=0\). For these reasons, we ignore the $c_0$ term in what follows, as we did in Section \ref{sec:qm}.

\subsection{Coincident-Point Limits}\label{sec:two_point}

It is well-known that ultraviolet divergences are a quite general feature of quantum field theories, with rare exceptions such as topological QFTs that lack local degrees of freedom. Perhaps the simplest example of such a divergence is that of the two-point function at coincident points; if a quantum field $\phi$ were a true random function, this quantity would measure the variance of the would-be random variable $\phi ( x )$.

First let us review why two-point function divergences are an elementary consequence of a non-trivial K\"all\'en-Lehmann expression (\ref{main_kallen_lehmann}). 

\begin{theorem}
Consider a quantum field theory in $d \geq 2$ spacetime dimensions which admits a representation (\ref{main_kallen_lehmann}) for its two-point function with $\rho$ not identically vanishing. Consider any set of mollifiers (\ref{mollifier_defn}) which approximate the Dirac delta function $\delta^{(d)} ( x - y )$ as $\epsilon \to 0$. Then there exist constants $A, B$ such that, for small $\epsilon$,
\begin{align}\label{coincident_point_limit}
    C \left( \eta_{\epsilon, x} , \eta_{\epsilon, x} \right) \geq \begin{cases} A \log \left( \frac{1}{\epsilon} \right) - B &\text{ if } d = 2 \\ A \epsilon^{2-d} - B &\text{ if } d > 2 \end{cases} \, .
\end{align}
\end{theorem}
This is nothing but the familiar logarithmic divergence of scalar correlation functions in two spacetime dimensions and corresponding power-law divergence in higher dimensions. In particular, for any $d \geq 2$ and any non-trivial K\"all\'en-Lehmann spectral measure, the coincident-point limit (\ref{coincident_point_limit}) diverges as $\epsilon \to 0$. 

\begin{proof}
    Because $\rho$ is
    locally finite, there exists some $M > 0$ such that
    \begin{align}\label{finite_spectral_interval}
        \rho \left( [0,M^2] \right) < \infty \, .
    \end{align}
    Let the (finite, positive) value of the measure assigned to this interval be $m_M = \rho \left( [0,M^2] \right)$. The Fourier transform of the mollifier (\ref{mollifier_defn}) is
    \begin{align}
        \widehat{\eta_{\epsilon,x}} ( p )
        = e^{- i p \cdot x} \widehat{\eta} ( \epsilon p ) \, .
    \end{align}
    In particular, the phase associated with the point $x$ drops out of its
    absolute value. Substituting this expression into the K\"all\'en-Lehmann
    representation (\ref{main_kallen_lehmann}), we find
    \begin{align}\label{mollified_KL}
        C \left( \eta_{\epsilon,x},\eta_{\epsilon,x} \right)
        =
        \int_0^\infty \rho ( d\mu^2 )
        \int_{\mathbb{R}^d}
        \frac{d^d p}{(2\pi)^d}
        \frac{\left| \widehat{\eta} ( \epsilon p ) \right|^2}
        {p^2+\mu^2} \, .
    \end{align}
    Since the integrand and the spectral measure are non-negative, we may
    bound the $\mu^2$ integral by restricting to the interval $[0,M^2]$. On this interval,
    \begin{align}
        \frac{1}{p^2+\mu^2}
        \geq
        \frac{1}{p^2+M^2} \, ,
    \end{align}
    and therefore
    \begin{align}\label{spectral_lower_bound}
        C \left( \eta_{\epsilon,x},\eta_{\epsilon,x} \right)
        \geq
        m_M
        \int_{\mathbb{R}^d}
        \frac{d^d p}{(2\pi)^d}
        \frac{\left| \widehat{\eta} ( \epsilon p ) \right|^2}
        {p^2+M^2} \, .
    \end{align}
    The normalization of the mollifier implies
    $\widehat{\eta} ( 0 ) = \int d^d x \, \eta ( x ) = 1$. By continuity of
    $\widehat{\eta}$, there consequently exists a constant $\delta>0$ such
    that
    \begin{align}
        \left| \widehat{\eta} ( k ) \right|^2
        \geq \frac{1}{2}
        \qquad \text{whenever} \qquad |k| \leq \delta \, .
    \end{align}
    It follows that $\left| \widehat{\eta} ( \epsilon p ) \right|^2
    \geq \frac{1}{2}$ throughout the momentum-space ball
    $|p| \leq \frac{\delta}{\epsilon}$. Thus
    \begin{align}\label{momentum_ball_lower_bound}
        C \left( \eta_{\epsilon,x},\eta_{\epsilon,x} \right)
        &\geq
        \frac{m_M}{2}
        \int_{|p|\leq \frac{\delta}{\epsilon}}
        \frac{d^d p}{(2\pi)^d}
        \frac{1}{p^2+M^2}
        \nonumber \\
        &=
        K_d
        \int_0^{\delta/\epsilon}
        dr \,
        \frac{r^{d-1}}{r^2+M^2} \, ,
    \end{align}
    where we have introduced the constant
    \begin{align}
        K_d \equiv
        \frac{m_M \operatorname{vol} ( S^{d-1} )}
        {2 (2\pi)^d}
        >0 \, .
    \end{align}
    
    We now distinguish the cases $d=2$ and $d>2$. In two dimensions, the
    remaining radial integral can be evaluated explicitly:
    \begin{align}
        \int_0^{\delta/\epsilon}
        dr \, \frac{r}{r^2+M^2}
        &=
        \frac{1}{2}
        \log \left(
        1+\frac{\delta^2}{M^2\epsilon^2}
        \right)
        \nonumber \\
        &\geq
        \log \left( \frac{\delta}{M\epsilon} \right) \, .
    \end{align}
    Therefore, after absorbing the finite constant
    $K_2 \log ( \delta/M )$ into an $\epsilon$-independent constant $B$, we
    obtain
    \begin{align}
        C \left( \eta_{\epsilon,x},\eta_{\epsilon,x} \right)
        \geq
        A \log \left( \frac{1}{\epsilon} \right) - B
        \qquad (d=2)
    \end{align}
    for some $A>0$.

    For $d>2$, take $\epsilon$ sufficiently small that
    $\frac{\delta}{\epsilon} \geq M$. Since
    $r^2+M^2 \leq 2r^2$ whenever $r\geq M$, equation
    (\ref{momentum_ball_lower_bound}) gives
    \begin{align}
        \int_0^{\delta/\epsilon}
        dr \, \frac{r^{d-1}}{r^2+M^2}
        &\geq
        \int_M^{\delta/\epsilon}
        dr \, \frac{r^{d-1}}{r^2+M^2}
        \nonumber \\
        &\geq
        \frac{1}{2}
        \int_M^{\delta/\epsilon}
        dr \, r^{d-3}
        \nonumber \\
        &=
        \frac{1}{2(d-2)}
        \left[
        \delta^{d-2}\epsilon^{2-d}-M^{d-2}
        \right] \, .
    \end{align}
    Absorbing the final, $\epsilon$-independent term into $B$ therefore yields
    \begin{align}
        C \left( \eta_{\epsilon,x},\eta_{\epsilon,x} \right)
        \geq
        A \epsilon^{2-d}-B
        \qquad (d>2)
    \end{align}
    for another constant $A>0$. This completes the second half of (\ref{coincident_point_limit}).
\end{proof}

We mention in passing that the same argument also shows directly that the momentum-space K\"all\'en-Lehmann function (\ref{momentum_space_KL}), which we repeat for convenience,
\begin{align}
G ( p ) = \int_0^\infty \frac{\rho ( d \mu^2 )}{p^2 + \mu^2} \, ,
\end{align}
is not integrable over momentum space. Indeed, using the finite interval of positive spectral weight introduced above, with $m_M = \rho ( [0,M^2] ) > 0$, positivity gives
\begin{align}
G ( p )
\geq
\int_{[0,M^2]}
\frac{\rho ( d \mu^2 )}{p^2+\mu^2}
\geq
\frac{m_M}{p^2+M^2} \, .
\end{align}
Consequently,
\begin{align}\label{G_integral_diverges}
\int_{\mathbb{R}^d} d^d p \, G ( p )
&\geq
m_M \operatorname{vol} \left( S^{d-1} \right)
\int_0^\infty dr \,
\frac{r^{d-1}}{r^2+M^2}
= \infty
\qquad ( d \geq 2 ) \, .
\end{align}
The final radial integrand behaves as $r^{d-3}$ at large $r$, producing a logarithmic divergence (for $d=2$) and a power-law divergence (for $d>2$). Thus the divergence of $\int d^d p \, G(p)$ is simply the unsmeared version of the same ultraviolet estimate responsible for the coincident-point limit above.

As the simple arguments above make clear, any putative NN-FT realization of a $d \geq 2$ theory with a non-trivial K\"all\'en-Lehmann spectral measure which lacks the coincident-point divergences ensured by (\ref{coincident_point_limit}) must violate at least one of the Osterwalder-Schrader assumptions of temperedness, Euclidean invariance, and reflection positivity. More precisely, any pair $(\phi_\theta,P(\theta))$ with bounded mollified coincident variance,
\begin{align}
    \lim_{\epsilon \to 0} \langle \phi_\theta ( \eta_{\epsilon, x} ) \phi_\theta ( \eta_{\epsilon, x} ) \rangle < \infty \, ,
\end{align}
is incompatible with a non-trivial K\"all\'en-Lehmann representation. This conclusion is general and architecture-independent, modulo the caveats that we have already mentioned, such as working with centered fields.

It is potentially more interesting to consider conditions under which such $N < \infty$ architectures have this finite smeared variance property, at least at a mathematical level, since operationally any finite number of random network evaluations on a computer will have finite variance.

A short estimate is sufficient to cover a broad class of finite-width architectures. We continue to work with the centered field, and assume that each parameter draw defines an ordinary function which can be integrated against the mollifier. Suppose that the pointwise variance is bounded in some compact neighborhood \(K\) of \(x\),
\begin{align}\label{locally_bounded_variance}
    V_K
    \equiv
    \sup_{y\in K}
    \mathbb{E}\!\left[
        \left|\phi_\theta(y)\right|^2
    \right]
    <\infty \, .
\end{align}
For all sufficiently small \(\epsilon\), the support of
\(\eta_{\epsilon,x}\) lies inside \(K\). Since the mollifier is non-negative
and has unit integral, the Cauchy--Schwarz inequality gives
\begin{align}\label{finite_mollified_variance_bound}
    C\left(\eta_{\epsilon,x},\eta_{\epsilon,x}\right)
    &=
    \mathbb{E}\!\left[
        \left|
        \int_{\mathbb{R}^d}d^d y\,
        \eta_{\epsilon,x}(y)\phi_\theta(y)
        \right|^2
    \right]
    \nonumber\\
    &\leq
    \int_{\mathbb{R}^d}d^d y\,
    \eta_{\epsilon,x}(y)
    \mathbb{E}\!\left[
        \left|\phi_\theta(y)\right|^2
    \right]
    \leq V_K \, .
\end{align}
The right side is independent of \(\epsilon\), so the mollified coincident
variance cannot diverge as \(\epsilon\to0\). In a Euclidean-invariant
ensemble the pointwise variance is independent of position, and the
assumption (\ref{locally_bounded_variance}) reduces simply to finite
variance at any one point.

To see that this condition includes familiar finite-width networks,
consider
\begin{align}
    \phi_\theta(x)
    =
    c+\sum_{i=1}^{N}
    a_i\sigma_i(w_i\cdot x+b_i) \, .
\end{align}
It is enough that \(c\) have finite second moment and that each of the
finitely many neuron outputs have second moment bounded on \(K\). Indeed,
\begin{align}
    \sup_{y\in K}
    \mathbb{E}\!\left[
        \left|\phi_\theta(y)\right|^2
    \right]
    \leq
    (N+1)
    \left(
        \mathbb{E}[|c|^2]
        +
        \sum_{i=1}^{N}
        \sup_{y\in K}
        \mathbb{E}\!\left[
            \left|
            a_i\sigma_i(w_i\cdot y+b_i)
            \right|^2
        \right]
    \right)
    <\infty \, .
\end{align}
No independence assumption is required. This condition holds, for
instance, for bounded activations and finite-variance output weights. It
also holds for ReLU and other activations of at most polynomial growth when
the parameter density has the corresponding finite moments, including the
usual Gaussian parameter densities. The same reasoning applies to any
fixed finite-depth architecture whose output obeys the local moment bound
(\ref{locally_bounded_variance}).

Consequently, every finite network in this broad finite-moment class has
bounded mollified coincident variance and cannot reproduce the non-trivial
K\"all\'en-Lehmann behavior of
(\ref{coincident_point_limit}). The finite-\(N\) qualification alone is not
sufficient: heavy-tailed parameter densities, singular
parameter-dependent amplitudes, or regulator limits can violate
(\ref{locally_bounded_variance}) even for a single neuron. We now turn to
explicit examples of these possibilities.

\subsection{Infinite-Variance Examples}\label{sec:counterexample}

It is instructive to clarify the role of NN-FT architectures which violate the local finite-variance condition of the preceding subsection, since it illustrates that a finite, or even single-neuron, network can produce a two-point function with coincident-point divergences, if one chooses an architecture and parameter density which violate this condition. In particular, such a finite-width network can reproduce a reflection-positive $2$-point function. 

There are several ways that one might do this. One is implicit in previous studies of the cos-net architecture which realizes the free scalar \cite{Halverson:2021aot,Frank:2026bui}. Here we specialize to $d=2$. The architecture takes the form
\begin{align}\label{free_scalar_cosnet_arch}
    \phi_\theta ( x ) = \frac{C}{\sqrt{N}} \sum_{i=1}^{N} \frac{a_i}{|W_i|} \cos \left( W_i \cdot x + c_i \right) \, , \qquad C^2 = \frac{\alpha' \left( \Lambda^2 - \epsilon^2 \right)}{\sigma_a^2}
\end{align}
with $a_i$ Gaussian with variance $\sigma_a^2$, $c_i$ uniform, and $W_i$ chosen uniform on the annulus
\begin{align}
    \epsilon < | W_i | < \Lambda \, .
\end{align}
Here $\epsilon$ plays the role of an IR regulator whereas $\Lambda$ is a UV cutoff. For finite $\Lambda$, this architecture satisfies the local finite-variance condition (\ref{locally_bounded_variance}). On any compact set $K \subset \mathbb{R}^d$,
\begin{align}
    \mathop{\mathrm{sup}}\limits_{x \in K} \left| \frac{C a_i}{\sqrt{N} | W_i |} \cos \left( W_i \cdot x + c_i \right) \right| \leq \frac{C | a_i |}{\sqrt{N} \epsilon} \, ,
\end{align}
and thus in the notation of the preceding subsection, one has
\begin{align}
    \mathop{\mathrm{sup}}\limits_{x \in K} | \phi_\theta ( x ) | \leq \frac{C}{\epsilon \sqrt{N}} \sum_{i=1}^{N} | a_i | \, ,
\end{align}
and hence the local finite-variance condition (\ref{locally_bounded_variance}) holds for every $0 < \epsilon < \Lambda < \infty$. Thus this architecture has finite variance and lacks coincident point divergences, which is manifest as this variance can be computed straightforwardly:
\begin{align}\label{exact_variance}
    \mathbb{E} \left[ \phi ( x )^2 \right] = \alpha' \log \left( \frac{\Lambda}{\epsilon} \right) \, .
\end{align}
However, the exact expression (\ref{exact_variance}) also makes it clear that this finite variance is a regulator effect. As either $\epsilon \to 0$ or $\Lambda \to \infty$, the variance (\ref{exact_variance}) diverges, and the architecture formally violates the local finite-variance condition (\ref{locally_bounded_variance}), even at finite $N$. The mechanisms by which these two limits produce infinite variance are different. In the $\Lambda \to \infty$ limit, one formally wishes to draw momentum from a uniform distribution on an annulus of infinite size, which is an improper prior that na\"ively assigns equal weight to arbitrarily large momentum modes. In the $\epsilon \to 0$ limit for finite $\Lambda$, in contrast, the parameter density remains normalizable but the variance divergence is due to the expectation value
\begin{align}
    \mathbb{E} \left[ \frac{1}{|W|^2} \right] = \frac{2}{\Lambda^2 - \epsilon^2} \int_{\epsilon}^{\Lambda} \frac{dr}{r} \, ,
\end{align}
which is logarithmically divergent as $\epsilon \to 0$. Thus, although this architecture reproduces the standard free boson propagator at finite $N$ for separations $\frac{1}{\Lambda} \ll r \ll \frac{1}{\epsilon}$, one must be careful about certain order-of-limits issues regarding the regulators.

A related single-neuron construction is the following. Suppose that we wish to realize the two-point function of a theory with a K\"all\'en-Lehmann representation
\begin{align}\label{target_rep_to_match}
    G ( p ) = \int_0^\infty \frac{\rho ( d \mu^2 )}{p^2 + \mu^2} \, ,
\end{align}
with the properties that $G ( p ) < \infty$ and that
\begin{align}\label{finite_KL_assumption}
    \int \frac{d^d p}{(2 \pi)^d} G ( p ) | \widehat{f} ( p ) |^2 < \infty
\end{align}
for the Fourier transform $\widehat{f} ( p )$ of any $f \in \mathcal{S} ( \mathbb{R}^d )$. These conditions are satisfied, for instance, by a massive free scalar field in $d \geq 2$ dimensions. More generally, by Tonelli's theorem, the condition (\ref{finite_KL_assumption}) holds for any theory whose distributional two-point function $C ( f, g ) = \mathbb{E} \left[ \overline{\phi ( f )} \phi ( g ) \right]$ obeys $C ( f, f ) < \infty$ for all $f \in \mathcal{S} ( \mathbb{R}^d )$. We note that this finiteness condition is distinct from the existence of coincident-point divergences, which involve limits $C ( \eta_{\epsilon, x} , \eta_{\epsilon, x} )$ as $\epsilon \to 0$, as discussed around (\ref{coincident_point_limit}).

Now consider the single-neuron architecture
\begin{align}\label{single_layer_example}
    \phi_\theta ( x ) = A ( W ) \cos ( W \cdot x + b ) \, , \qquad A ( p ) = \sqrt{ \frac{2 G ( p ) }{( 2 \pi )^d q ( p ) }} \, ,
\end{align}
where $\theta \in \{ W, b \}$, $W$ is drawn from any even, strictly positive probability distribution $q ( W )$ on $\mathbb{R}^d$ and $B$ is drawn uniformly on $[0, 2 \pi]$. This simple example has two properties.

\begin{enumerate}[label = (\alph*)]
\item\label{single_neuron_divergences} First, it is able to reproduce coincident-point divergences. Indeed, beginning from
\begin{align}
    \mathbb{E} \left[ \phi_\theta ( x )^2 \right] = \int d \theta \, P ( \theta ) \, A ( W )^2 \cos^2 \left( W \cdot x + b \right) \, ,
\end{align}
averaging over the uniform phase $B$ gives the average value $\mathbb{E}_B \left[ \cos^2 \left( W \cdot x + b \right) \right] = \frac{1}{2}$, and we are left with the remaining expectation over $W$,
\begin{align}
    \mathbb{E} \left[ \phi_\theta ( x )^2 \right] &= \frac{1}{2} \int d^d W \, q ( W ) \, \frac{2 G ( W ) }{( 2 \pi )^d q ( W ) } \nonumber \\
    &= \int \frac{d^d W}{(2 \pi)^d} G ( W ) = \infty \, ,
\end{align}
precisely because of the infinite integral of the K\"all\'en-Lehmann representation discussed around equation (\ref{G_integral_diverges}).

\item Second, this architecture can match the two-point function of a
``target'' theory with KL representation (\ref{target_rep_to_match}). Let
$f,g\in\mathcal{S}(\mathbb{R}^d)$. With the Fourier-transform convention
used above, the smeared neuron is
\begin{align}
    \phi_{\theta}(f)
    =
    \frac{A(W)}{2}
    \left[
        e^{iB}\widehat{f}(-W)
        +
        e^{-iB}\widehat{f}(W)
    \right] \, .
\end{align}
Averaging over the uniform phase removes the terms proportional to
$e^{\pm 2iB}$, so that
\begin{align}
    \mathbb{E}
    \left[
        \overline{\phi_\theta (f)} \phi_\theta (g)
    \right]
    &=
    \frac{1}{2}
    \int_{\mathbb{R}^d}
    \frac{d^d W}{(2\pi)^d}
    G(W)
    \left[
        \overline{\widehat{f}(-W)}\widehat{g}(-W)
        +
        \overline{\widehat{f}(W)}\widehat{g}(W)
    \right]
    \nonumber\\
    &=
    \int_{\mathbb{R}^d}
    \frac{d^d p}{(2\pi)^d}
    G(p)\,
    \overline{\widehat{f}(p)}\widehat{g}(p) \, .
\end{align}
Here we used the evenness of $G$, which follows from
(\ref{target_rep_to_match}), and changed variables $W\to-W$ in the first
term. Substituting the KL representation of $G$ then gives
\begin{align}
    \mathbb{E}
    \left[
        \overline{\phi_\theta (f)}\phi_\theta (g)
    \right]
    =
    \int_0^\infty \rho(d\mu^2)
    \int_{\mathbb{R}^d}
    \frac{d^d p}{(2\pi)^d}
    \frac{
        \overline{\widehat{f}(p)}\widehat{g}(p)
    }{p^2+\mu^2}
    =
    C(f,g) \, .
\end{align}
Thus the single-neuron network is able to correctly match the smeared two-point functions of a theory with a given KL representation, while preserving the divergent variances of point \ref{single_neuron_divergences} when the functions $f$ and $g$ approach coincident delta functions.

\end{enumerate}

The lesson of the two examples (\ref{free_scalar_cosnet_arch}) and (\ref{single_layer_example}) is the following: it \emph{is} possible to realize the two-point function of a reflection-positive quantum field theory with a finite-width (or even width-$1$) architecture, if one is willing to relax the local finite-variance condition (\ref{locally_bounded_variance}) presented earlier in this Section. This is related to the observation, first made in \cite{Halverson:2021aot}, that one is free to adjust the functional dependence of output weights on input weights in a cos-net architecture in order to tune the power spectrum of the theory. Of course, as we have already emphasized, practical Monte Carlo sampling on a computer always produces finite-variance random variables. There is therefore a second, operational question of whether using finite-variance approximations to architectures with infinite pointwise variance -- perhaps along with suitable regularization and mollification -- can be useful in practice. It appears that this is the case, at least in NN-FT simulations of Liouville \cite{Ferko:2026axm} and Maxwell \cite{toappear} theory, but we leave a more systematic investigation of this question to future work.

Let us conclude this subsection with a quantum-mechanical analogue which \emph{does} have finite variance, but illustrates a related point. We mentioned that, in Section \ref{sec:qm}, we restricted to fixed-feature architectures that do not have random input weights, and showed that such finite-$N$ architectures are incompatible with the $1d$ K\"all\'en-Lehmann representation. By repeating the single-neuron construction (\ref{single_layer_example}) with an appropriate KL representation for a $1d$ theory, it is again possible to reproduce a desired (reflection-positive) two-point function, although generically higher-point functions will not be reflection positive and/or the theory will violate cluster decomposition. For instance, choosing the architecture $x(t) = \sqrt{2 c} \cos ( W t + b )$ with $b \sim U ( [ 0, 2 \pi ] )$ and $W \sim \frac{m d W}{\pi ( W^2 + m^2 )}$ reproduces the harmonic oscillator two-point function $\langle x ( t ) x ( s ) \rangle = c e^{- m | t - s |}$, but not its higher-point functions.

\subsection{Momentum-Space Analysis}\label{sec:momentum_space}

We have just seen in Section \ref{sec:counterexample} that a finite-width network can reproduce a reflection-positive K\"all\'en-Lehmann two-point function if one allows sufficiently singular parameter statistics. In particular, the width-$1$ architecture (\ref{single_layer_example}) distributes the momentum $W$ over the ensemble so that its \emph{averaged} covariance has support throughout momentum space, even though each individual field draw contains only the two momenta $\pm W$. The two-point function alone does not distinguish these two notions of momentum support. We will now show that cluster decomposition supplies precisely the additional input that is needed: in a clustered QFT, every open momentum region occurs in a typical field draw.

We first recall a probabilistic consequence of clustering which will be important below. Let $\tau_y \phi$ denote the field translated by $y$. We will use a strong version of the cluster axiom related to the notion referred to as \emph{mixing} in the stochastic process literature,
\begin{align}\label{cluster_decomposition_mixing}
    \mathbb{E}\left[
        F(\phi) G(\tau_y \phi)
    \right]
    \longrightarrow
    \mathbb{E}\left[F(\phi)\right]
    \mathbb{E}\left[G(\phi)\right]
    \qquad \text{as } |y| \to \infty \, ,
\end{align}
for bounded observables $F$ and $G$ of a field draw. Physically, this is essentially the statement that interactions (and thus correlations) grow weak at large spatial separations; mathematically, one says that the field measure is \emph{ergodic under translations}.

This clustering property implies the following simple fact about translation-invariant events.\footnote{In the language of measure theory, by ``event'' we mean an element of the $\sigma$-algebra of measurable sets.} If $A$ is a translation-invariant event, we may take $F=G=\mathbf{1}_A$ in (\ref{cluster_decomposition_mixing}) to obtain
\begin{align}
    \mathbb{P}(A)
    =
    \mathbb{P}(A)^2 \, .
\end{align}
Therefore $\mathbb{P}(A)$ is either $0$ or $1$. Said differently, a theory exhibiting cluster decomposition cannot involve a non-trivial probability distribution over sectors distinguished by translation-invariant data; each translation-invariant event is ``all-or-nothing''. This is the only consequence of cluster decomposition that we will use in this subsection.

We now apply this observation to the momentum content of a field configuration. As in (\ref{main_kallen_lehmann}), we work with the centered field; subtracting the constant one-point function can only change the contribution at zero momentum. If a network reproduced the QFT correlation functions, its one-point function would be the same constant, so we may center the network at the same time without changing the support argument.

\begin{lemma}\label{full_momentum_support}
Let $\phi$ be a centered, translation-ergodic random Schwartz distribution whose connected two-point function has the non-trivial K\"all\'en-Lehmann form (\ref{main_kallen_lehmann}). Then
\begin{align}\label{typical_full_momentum_support}
    \operatorname{supp}\widehat{\phi}
    =
    \mathbb{R}^d
    \qquad \text{with probability one} \, .
\end{align}
\end{lemma}
%
% \aaron{AM: \st{The conclusion can be stated more intuitively -- but somewhat less precisely -- in words.}} 
Intuitively, for any QFT obeying all of the Osterwalder-Schrader axioms (subject to the usual caveats about centering, non-trivial $\rho$, etc.), given any momentum vector $p \in \mathbb{R}^d$, a typical random field draw will always have a momentum distribution including \emph{some} continuum of momenta near $p$. Here one must be slightly careful: because the Fourier transform is an automorphism of $\mathcal{S} ( \mathbb{R}^d )$ onto itself, duality allows the Fourier transform to be defined on all Schwartz distributions in $\mathcal{S}' ( \mathbb{R}^d )$, but the result is typically another distribution rather than a genuine function. A simple example illustrating this point is that the Fourier transform of a cosine is a sum of delta distributions. The more precise statement is that the Fourier transform of any fully-OS random field draw \emph{when smeared against an appropriate test function localized near a given momentum} is non-zero.

\begin{proof}
Fix a non-empty open ball $U \subset \mathbb{R}^d$ and consider the event
\begin{align}\label{missing_momentum_event}
    A_U
    =
    \left\{
        \widehat{\phi}
        \text{ vanishes as a distribution on } U
    \right\} \, .
\end{align}
This is a measurable event; for instance, it can be written using a countable separating family of smooth test functions supported in $U$. A translation acts in momentum space by multiplication with a phase,
\begin{align}\label{translation_momentum_phase}
    \widehat{\tau_y \phi}(p)
    =
    e^{- i p \cdot y}\widehat{\phi}(p) \, .
\end{align}
As this phase is nowhere vanishing, translations cannot create or remove Fourier support. Thus $A_U$ is translation invariant, and ergodicity implies that $\mathbb{P}(A_U)$ is either $0$ or $1$.

We next show that this probability cannot be $1$. Choose a non-zero test function $f \in \mathcal{S}(\mathbb{R}^d)$ whose Fourier transform has compact support in $-U$, with the support chosen away from $p=0$. Up to the harmless reflection associated with the Fourier-transform convention, $\phi(f)$ pairs $\widehat{\phi}$ with a test function supported in $U$. If $A_U$ occurred almost surely, then $\phi(f)$ would vanish almost surely and would therefore have zero variance. The possible additional zero-momentum term discussed after (\ref{main_kallen_lehmann}) makes no contribution because $\widehat f(0)=0$, while the propagating part gives
\begin{align}\label{positive_momentum_variance}
    \mathbb{E}\left[|\phi(f)|^2\right]
    & =
    \int_0^\infty \rho(d\mu^2)
    \int_{\mathbb{R}^d}
    \frac{d^d p}{(2\pi)^d}
    \frac{|\widehat{f}(p)|^2}{p^2+\mu^2}
    \nonumber \\
    & > 0 \, .
\end{align}
The strict inequality follows because $\rho$ is a non-zero positive measure and $\widehat f$ is not identically zero. We conclude that $\mathbb{P}(A_U)=0$.

Finally, take the countable collection of momentum-space balls with rational centers and rational radii. With probability one, $\widehat{\phi}$ vanishes on none of these balls. Since every non-empty open set contains such a ball, this is equivalent to (\ref{typical_full_momentum_support}).
\end{proof}

The role of clustering in Lemma \ref{full_momentum_support} is worth emphasizing. The K\"all\'en-Lehmann two-point function implies only that a given momentum region cannot be absent from \emph{every} draw. Translation ergodicity strengthens this statement by showing that the region is absent from almost no draws. This is precisely the step which fails for the random plane-wave ensemble (\ref{single_layer_example}), where the sampled momentum $W$ is a translation-invariant hidden label.

Let us now compare this conclusion with the single-layer architecture (\ref{single_layer}),
\begin{align}\label{finite_momenta_example}
    \phi_\theta(x)
    =
    \sum_{i=1}^{N}
    w_i \sigma_i(a_i \cdot x+b_i) \, .
\end{align}
Here $a_i \in \mathbb{R}^d$ is the vector of input weights, while $w_i$ is the scalar output weight. We assume only that the activations are locally integrable and grow at most polynomially, so that the neurons define Schwartz distributions. These mild conditions include the usual bounded, polynomial, and ReLU-type activations; no independence or finite-moment assumption on the parameters is required.

The main physical observation is that any such field draw of the architecture (\ref{finite_momenta_example}) has momenta which are supported at finitely many lines, proportional to the drawn values of the input weights $a_i$. This is unlike a field draw from a fully-OS quantum field theory, and is formalized in the following result.

\begin{lemma}\label{ridge_momentum_lines}
For every fixed parameter draw in (\ref{single_layer}),
\begin{align}\label{finite_network_momentum_support}
    \operatorname{supp}\widehat{\phi_\theta}
    \subseteq
    \bigcup_{i=1}^{N} \mathbb{R}a_i \, .
\end{align}
\end{lemma}

\begin{proof}
Consider first a single neuron with $a_i \neq 0$, and set $e_i=a_i/|a_i|$. We decompose position and momentum into longitudinal and transverse components,
\begin{align}
    x=s e_i+y \, , \qquad
    p=q e_i+p_\perp \, , \qquad
    y \cdot e_i=p_\perp \cdot e_i=0 \, .
\end{align}
The neuron depends only on the longitudinal coordinate $s$ and is constant in the $d-1$ transverse directions. Defining $h_i(s)=\sigma_i(|a_i|s+b_i)$, its distributional Fourier transform therefore takes the form
\begin{align}\label{ridge_neuron_fourier_transform}
    \widehat{\sigma_i(a_i \cdot x+b_i)}
    (q e_i+p_\perp)
    =
    (2\pi)^{d-1}
    \widehat{h_i}(q)
    \delta^{(d-1)}(p_\perp) \, .
\end{align}
The transverse delta function restricts the support to the line $\mathbb{R}a_i$. If $a_i=0$, the neuron is constant and its Fourier support is contained in $\{0\}=\mathbb{R}a_i$. Multiplication by the output weight $w_i$ cannot enlarge this support, while the Fourier support of a finite sum is contained in the union of the supports of its terms. This proves (\ref{finite_network_momentum_support}).
\end{proof}

We can now combine these two observations to obtain the desired result.

\begin{theorem}[Momentum-space obstruction]\label{thm:momentum_space_no_go}
Let $d\geq2$, and let $\phi$ be a non-trivial real scalar generalized random field satisfying the Osterwalder--Schrader axioms, with cluster decomposition understood in the mixing form (\ref{cluster_decomposition_mixing}). Suppose that its connected two-point function has the K\"all\'en-Lehmann form (\ref{main_kallen_lehmann}) with $\rho$ not identically zero. Then no finite-width architecture of the form (\ref{single_layer}), with activations satisfying the mild temperedness conditions stated above, can reproduce all of the correlation functions of $\phi$. This conclusion holds for an arbitrary joint probability law of the finitely many network parameters.
\end{theorem}

\begin{proof}
For $d \geq 2$, a finite union of one-dimensional lines is a proper closed subset of $\mathbb{R}^d$ and therefore misses some open momentum-space ball. Lemmas \ref{full_momentum_support} and \ref{ridge_momentum_lines} consequently show that a typical QFT draw has full Fourier support, whereas every finite-network draw fails to have this property. The two measures on field configurations therefore cannot agree.

Although this already establishes a disagreement between the full field-draw laws, our definition of a neural network realization is phrased in terms of correlation functions. It is therefore useful to see the obstruction directly in a finite correlator. Fix the proposed width $N$ and choose $N+1$ symmetric open momentum regions $V_k=-V_k$, with closures away from the origin, such that every line through the origin intersects at most one $V_k$. Such a choice is possible in every dimension $d\geq2$. For each $k$, choose a countable separating family of real test functions $f_{k,\ell}$ whose Fourier transforms have compact support in $V_k$.

By Lemma \ref{full_momentum_support}, the restriction of $\widehat{\phi}$ to every $V_k$ is non-zero almost surely. Therefore, for almost every QFT draw and each $k$, there is at least one $\ell$ for which $\phi(f_{k,\ell})\neq0$. The countable union, over tuples $(\ell_1,\ldots,\ell_{N+1})$, of the events
\begin{align}
    \bigcap_{k=1}^{N+1}
    \left\{
        \phi(f_{k,\ell_k})\neq0
    \right\}
    \, ,
\end{align}
has probability one. At least one fixed tuple must consequently have positive probability. Selecting this tuple and writing its test functions as $f_1,\ldots,f_{N+1}$, we find
\begin{align}\label{qft_sector_correlation}
    \mathbb{E}\left[
        \prod_{k=1}^{N+1}\phi(f_k)^2
    \right]
    >0 \, .
\end{align}
Here the expectation is finite by the usual OS regularity assumption for the corresponding smeared Schwinger function.

For a width-$N$ network, however, the $N$ lines in (\ref{finite_network_momentum_support}) can intersect at most $N$ of the $N+1$ regions $V_k$. Hence at least one of the smeared fields $\phi_\theta(f_k)$ vanishes for every parameter draw, although the identity of this field may depend on the draw. It follows that
\begin{align}\label{network_sector_correlation}
    \mathbb{E}\left[
        \prod_{k=1}^{N+1}\phi_\theta(f_k)^2
    \right]
    =0 \, .
\end{align}
Equations (\ref{qft_sector_correlation}) and (\ref{network_sector_correlation}) exhibit an explicit disagreement in a smeared $(2N+2)$-point function. This proves the theorem without any appeal to moment determinacy of the full field law.
\end{proof}

The momentum-space proof can be summarized in three steps: K\"all\'en-Lehmann positivity gives two-point fluctuations in every open momentum region; cluster decomposition forces every such region to occur in a typical draw; and a finite-$N$ architecture confines each draw to finitely many momentum lines. These statements are incompatible for $d\geq2$. Unlike the coincident-point argument, this conclusion does not require a finite-variance assumption on the parameter density.\footnote{It is, however, specific to finite single-layer networks and does not by itself extend to arbitrary deep architectures. We conjecture that similar results apply in that setting.} Nonetheless, Theorem \ref{thm:momentum_space_no_go} establishes that -- while the theory defined by a single-layer network (\ref{single_layer}) may be interesting -- at finite $N$, such a theory must necessarily give up one of the Osterwalder-Schrader axioms.

\subsection{Distributional Neurons}\label{subsec:distributional_neurons}

The preceding results illustrate a recurring theme of this article: what can be preserved at finite $N$ depends both on the physical property of interest and on the architecture used to represent the field. For instance, we saw that -- under a mild local finite-variance condition -- ordinary finite-width networks cannot reproduce the coincident-point behavior of two-point functions. However, sufficiently singular parameter statistics can evade that obstruction, as the width-$1$ construction (\ref{single_layer_example}) shows. For the single-layer architecture (\ref{single_layer}), however, imposing cluster decomposition produces a different obstruction in momentum space. It is therefore natural to ask whether the mismatch between finite networks and the OS properties of quantum fields can instead be avoided by changing the neurons themselves.

One might wonder whether it is possible to draw inspiration from the quantum-mechanical discussion of Section \ref{sec:qm}. There, a K\"all\'en--Lehmann representation of the covariance leads to a KKL expansion, which we repeat,
\begin{align}\label{kkl_repeated_later}
    x_t = \langle x ( t ) \rangle + \sum_{k=1}^{\infty} \theta_k e_k ( t ) \, , 
\end{align}
and that involves fixed deterministic neurons (or features) $e_k ( t )$ and random coefficients. Since a quantum field is generally a Schwartz distribution rather than an ordinary function, one might seek a QFT analogue in which the deterministic modes are themselves allowed to be elements of $\mathcal{S}'(\mathbb{R}^d)$. We will refer to these modes as \emph{distributional neurons}. Allowing such neurons removes the elementary regularity mismatch between the types of objects represented by random finite-$N$ neural networks (functions) and random field configurations (Schwartz distributions): now every network draw may itself be a proper Schwartz distribution. However, in spirit with the finite-$N$ focus of this article, we will again truncate the distributional sum inspired by (\ref{kkl_repeated_later}) to finite width. We ask only whether the natural finite, KKL-like truncation can reproduce the QFT covariance.

We will see that the answer is no. As the argument is structurally quite similar to that of Section \ref{nnqmDimCount}, we will present the obstruction only briefly. Consider a finite NN-FT representation of a field configuration
\begin{equation}\label{distributional_KKL}
    \Phi_{ \vartheta } = F_0\, +\, \sum_{ k = 1 }^{N} \vartheta_k F_k\,,
\end{equation}
where $ F_k \in \mathcal{S}' ( \mathbb{R}^d) $ are deterministic distributional neurons and $\vartheta_k$ are random variables. As before, each $F_k$ is a map from the infinite-dimensional space $\mathcal{S} ( \mathbb{R}^d )$ of test functions into $\mathbb{R}$ (or $\mathbb{C}$, for a complex scalar), and thus cannot be injective. In particular, each $F_k$ has an infinite-dimensional kernel, and there is a non-zero test function $f \in \mathcal{S} ( \mathbb{R}^d )$ which is annihilated by every distributional neuron: $F_k ( f ) = 0$ for all $0 \leq k \leq N$. Smearing (\ref{distributional_KKL}) against this test function then gives
\begin{align}
    \Phi_{ \vartheta } ( f ) = 0 \, ,
\end{align}
for every draw of the parameters $\vartheta$, and thus the variance of the random variable $\Phi_{ \vartheta } ( f )$ vanishes. But this is incompatible with a non-trivial K\"all\'en-Lehmann covariance, as we have already seen. Indeed, the positivity of the spectral measure in (\ref{main_kallen_lehmann}) gives
\begin{align}
    \mathbb{E} \left[ \overline{\Phi ( f )} \Phi ( f ) \right] = C(f,f)=\int_0^\infty \rho(d\mu^2)\int_{\mathbb{R}^d}\frac{d^dp}{(2\pi)^d}\frac{|\widehat f(p)|^2}{p^2+\mu^2}>0 \, ,
\end{align}
for all nonzero $f \in \mathcal{S} ( \mathbb{R}^d )$, as we have already seen around (\ref{positive_momentum_variance}).

We have therefore learned that distributional neurons remove the ``type-of-object'' mismatch, but not the finite-rank obstruction familiar from the NN-QM analysis. A finite collection of distributional features still leaves infinitely many non-zero directions in which the field is completely deterministic, whereas a field with non-trivial K\"all\'en-Lehmann representation fluctuates in every such direction. Consequently, no finite truncation of the KKL-type form (\ref{distributional_KKL}) can reproduce the correlation functions of a non-trivial QFT obeying temperedness, Euclidean invariance, and reflection positivity, since these three conditions imply the existence of such a K\"all\'en-Lehmann representation.

\section{Conclusion}\label{sec:conclusion}

In this article, we have studied features of neural network representations for quantum mechanical models and quantum field theories at finite width. The primary conclusion is that, for $N < \infty$, such neural network representations typically violate at least one of the standard assumptions in constructive formulations of quantum theories, such as reflection positivity or cluster decomposition. Morally, this violation stems from the fact that finite-$N$ architectures do not represent the objects on which typical path integral measures are supported. However, the analysis of this work has identified the specific hallmarks of conventional QFTs which are sacrificed by truncation to $N < \infty$.

There are at least two lessons that one might extract from this result. The first is that, in order to \emph{exactly} realize a conventional QM or QFT via a neural network representation, we must either take an infinite-width limit or use an architecture and parameter density which evade one of the assumptions of our theorems. The second is that theories defined by finite-$N$ neural networks may be interesting in their own right, as they are in some sense more ``exotic'' than familiar quantum models.

This work opens up several directions for future investigation. A natural question is whether our qualitative theorems can be promoted to quantitative bounds which characterize the degree to which one of the Osterwalder-Schrader axioms (such as reflection positivity) must be violated for a given architecture as a function of $N$. For instance, letting $F$ be a function of field insertions and $T$ be the reflection operator about a hyperplane, one might envision an inequality of the form
\begin{align}
    \left\langle F ( T F )^\ast \right\rangle \geq f ( N , F ) 
\end{align}
for some function $f$ depending upon the width of the architecture. Assuming that the architecture and parameter density under consideration reproduce a reflection positive QFT at infinite width, one has $\lim_{N \to \infty} f ( N , F ) = 0$ for every $F$, but it would be interesting to investigate ``how negative'' such a function $f$ could be at finite $N$ for realistic architectures such as a single-layer network. In Lorentzian language, such an inequality would characterize the degree to which unitarity is violated in a finite-width NN-QM or NN-FT. One might also consider a similar study of violations of cluster decomposition or Euclidean invariance.

A second direction is to classify which combinations of the Osterwalder-Schrader properties can be retained at finite $N$. For example, one may ask whether architectures can be constructed which preserve
Euclidean invariance and reflection positivity of all $n$-point functions exactly, but violate cluster decomposition. This would turn the present obstructions into a classification of finite-width theories according to the physical axiom that they sacrifice. It would also be interesting to characterize exceptional theories for which finite representations may exist, such as theories without local degrees of freedom and ``degenerate'' or topological theories.

A related question is whether one can formulate a more complexity-theoretic notion of a finite neural network realization. The existence of a formal one-parameter NN-FT representation establishes that parameter counting by itself cannot distinguish a genuine finite architecture such as (\ref{single_layer}) from an encoding which hides infinitely many operations in the digits of a real number. Rather than specifying the functional form of a finite-width architecture, one might instead impose more general requirements such as computability, bounded numerical precision, or explicit bounds on the number of permitted operations. Using such a definition, one might be able to demonstrate that any neural network representation with ``finite complexity'' -- in some suitable sense -- must also violate certain properties of standard QFTs, as in the examples considered in this work.

It would also be useful to extend the present analysis to theories which are \emph{not} defined by a real non-negative measure over field configurations, such as models with fermions or gauge theories with a theta angle. Quantum theories with a sign problem (i.e. a complex or sign-indefinite weight) fall outside the scope of the present investigation, although in principle their observables might still be computed from a neural network representation by splitting the measure into a positive-definite modulus along with a phase, and treating the phase as part of the observable. One might expect that finite-width neural network representations of theories with a sign problem must also violate certain physical assumptions, although this remains an open question.

Finally, it would be intriguing to investigate other properties of finite-width NN-FTs more systematically. Although we have established that such models necessarily lack typical properties such as Euclidean covariance, reflection positivity, or cluster decomposition, this does not mean that these theories are undeserving of study. For instance, translating to the Lorentzian versions of these properties, it is common to consider quantum field theories which lack Lorentz invariance (e.g. in condensed matter models) or unitarity (such as open quantum systems). Another important question is whether the non-local interactions observed in finite-$N$ truncations \cite{Demirtas:2023fir,Robinson:2025ybg} are universal, and whether their dependence on $N$ can be understood directly from the same spectral and infinite-dimensional arguments employed in this work. Answering these questions may help clarify both the limitations and the utility of neural networks as finite representations of continuum quantum systems.

\section*{Acknowledgements}

We are grateful to James Halverson and Thomas Harvey for fruitful discussions related to this work. We also acknowledge the language model ChatGPT-5.6 Sol, which was used in this research. The authors take full responsibility for the accuracy of this work, having reviewed, verified, and approved all AI-generated content. C.\,F. is supported by the National Science Foundation under Cooperative Agreement PHY-2019786 (the NSF AI Institute for Artificial Intelligence and Fundamental Interactions). A.\,M. is supported by the National Science Foundation grant PHY-2209903.

\appendix

\bibliographystyle{utphys}
\bibliography{master}
\end{document}